\documentclass[12pt]{iopart}

\usepackage{graphics}
\usepackage{graphicx}
\usepackage{color, hyperref}
\usepackage{cite}
\usepackage{amsthm}
\usepackage{amsmath, amssymb, latexsym, bm, mathtools, braket, multirow, url}

\usepackage{arydshln}
\usepackage{cancel}

\theoremstyle{plain}
\newtheorem{theorem}{Theorem}
\newtheorem{definition}[theorem]{Definition}

\newtheorem{corollary}[theorem]{Corollary}
\newtheorem{lemma}[theorem]{Lemma}
\newtheorem{postulate}[theorem]{Postulate}

\theoremstyle{definition}

\newcommand{\cH}{\mathcal{H}}

\newcommand{\cK}{\mathcal{K}}

\newcommand{\cM}{\mathcal{M}}

\newcommand{\cS}{\mathcal{S}}
\newcommand{\cT}{\mathcal{T}}

\newcommand{\cV}{\mathcal{V}}

\newcommand{\cX}{\mathcal{X}}

\newcommand{\ketbra}[2]{\ket{#1}\!\bra{#2}}
\newcommand{\abs}[1]{\left|#1\right|}

\DeclareMathOperator{\conv}{conv}

\DeclareMathOperator{\sco}{sc} 
\DeclareMathOperator{\nege}{neg} 

\begin{document}

\title{Perfect Discrimination in Approximate Quantum Theory of General Probabilistic Theories}

\author{Yuuya Yoshida$^1$, Hayato Arai$^1$, and Masahito Hayashi$^{2,1,3,4}$}

\address{$^1$ Graduate School of Mathematics, Nagoya University, Nagoya, Furo-cho, Chikusa-ku, 464-8602, Japan}
\address{$^2$ Shenzhen Institute for Quantum Science and Engineering, 
Southern University of Science and Technology, Nanshan District, Shenzhen, 518055, China}
\address{$^3$ Center for Quantum Computing, Peng Cheng Laboratory, Shenzhen, 518000, China}
\address{$^4$ Centre for Quantum Technologies, National University of Singapore, 
3 Science Drive 2, 117542, Singapore}
\eads{\mailto{m17043e@math.nagoya-u.ac.jp}, \mailto{m18003b@math.nagoya-u.ac.jp}, 
\mailto{hayashi@sustech.edu.cn}, \mailto{masahito@math.nagoya-u.ac.jp}}
\vspace{10pt}
\begin{indented}
\item[] April 2020
\end{indented}

\begin{abstract}
As a modern approach for the foundation of quantum theory,
existing studies of General Probabilistic Theories gave various models of states and measurements that are quite different from quantum theory.
In this paper, to seek a more realistic situation, we investigate models approximately close to quantum theory.
We define larger measurement classes that are smoothly connected with the class of POVMs via a parameter,
and investigate the performance of perfect discrimination.
As a result, we give a sufficient condition of perfect discrimination,
which shows a significant improvement beyond the class of POVMs.
\end{abstract}

\vspace{2pc}

\maketitle

\section{Introduction}
Quantum Theory (QT) is described by operators on Hilbert spaces, 
and the description is suitable to represent physical systems.
Many researchers have tried to give a foundation of the mathematical description.
A modern operational approach that starts with statistics of measurement outcomes 
is called General Probabilistic Theories (GPTs) 
\cite{Janotta2014,Lami2017,Short2010,Barnum2010,Dahlsten2012,Yoshida2018,Aubrun2018,
Bae2016,Matsumoto2018,Janotta2013,Muller2012,Muller2013,Masanes2011,Richens2017,Lee2015,PR1994,Plavala2017,Arai2019}.
Simply speaking, a GPT is defined by state/measurement classes 
that satisfy the following postulate: 
\begin{quote}
	\textit{Non-negativity of probability}: 
	For each measurement and each state, 
	the probability to obtain each measurement outcome is non-negative.
\end{quote}
In QT, the state class and measurement class are given as density matrices 
and Positive-Operator Valued Measures (POVMs) respectively, 
which indeed satisfy non-negativity of probability.
In this way, QT is a typical example of GPTs, 
and so is Classical Probability Theory (CPT).
Unfortunately, there is no operational reason in the sense of GPTs 
why only QT and CPT describe physical systems.
That is, no studies investigated how one denies an alternative realistic model of GPTs 
while it is known that there are superior models to QT/CPT with respect to information processing 
\cite{Lami2017,Aubrun2018,Lee2015,PR1994,Plavala2017,Arai2019}.

Preceding studies of GPTs defined models by restricting a state class to a much smaller one than QT/CPT.
Once restricting a state class, 
non-negativity of probability becomes a weaker condition, 
and the allowed measurement class becomes larger.
Consequently, measurement classes of preceding studies are much larger than QT/CPT,
and the classes sometimes show superiority of information processing.
For example, the PR box, which is defined by restricting states to only convex combinations of four states, 
violates Bell's inequality more strongly than QT, i.e., 
exceeds Tsirelson's bound \cite{PR1994,Plavala2017}. 
Also, Ref.\ \cite{Arai2019} focused on the case 
when available states are restricted to only separable states and 
all measurements with non-negativity of probability are allowed.
The pair of these state/measurement classes is called SEP, 
and Ref.\ \cite{Arai2019} showed that SEP has 
the superiority of perfect discrimination of bipartite separable pure states.

\begin{figure}[t]
	\centering
	\includegraphics[scale=0.35]{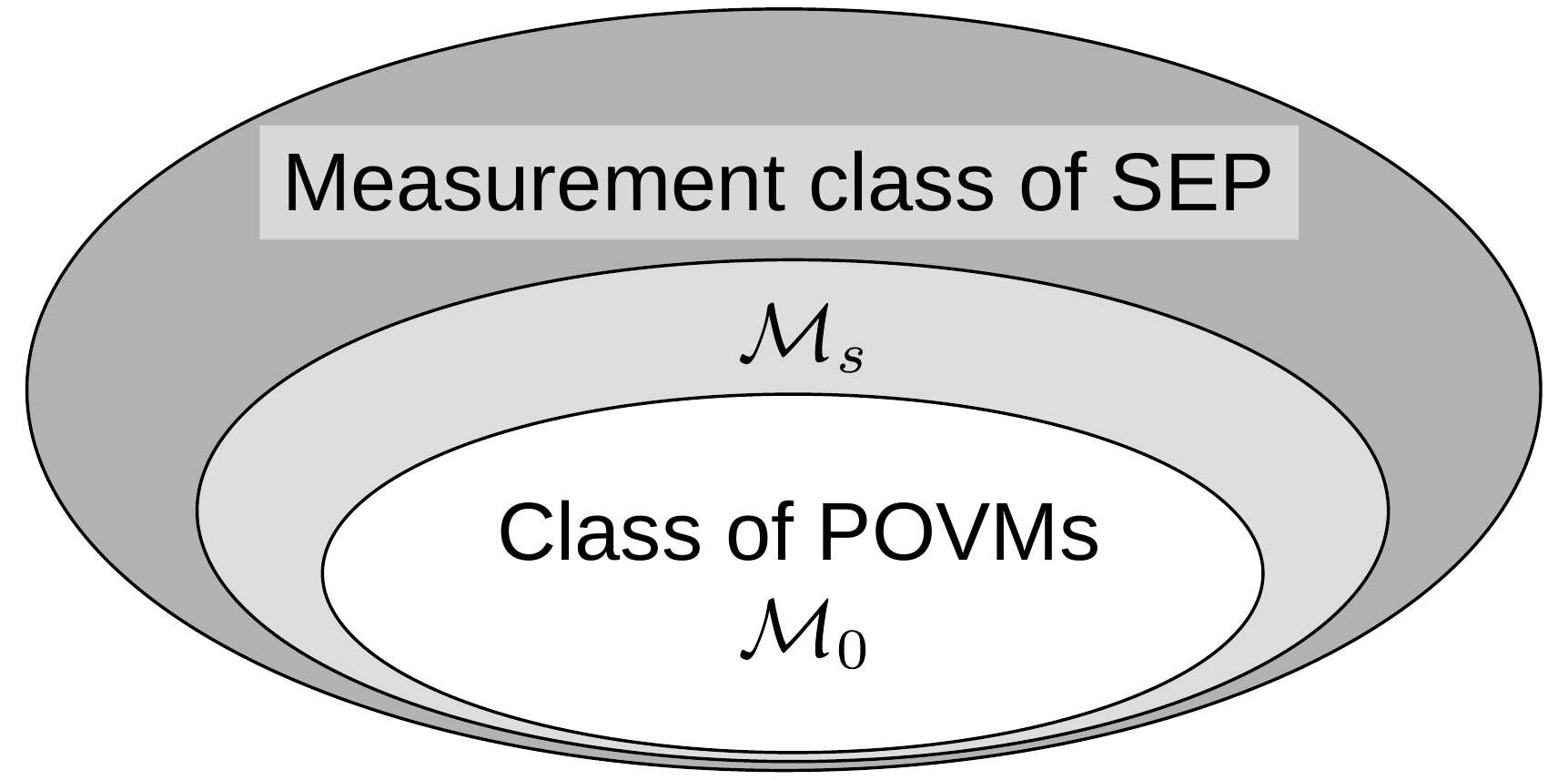}
	\caption{The measurement classes $\cM_s$ are smoothly connected 
	with the class $\cM_0$ of POVMs via the parameter $s$.}
	\label{fig:classes}
\end{figure}

However, since the above models are too far from QT, 
the reality of these models is easily denied. 
Hence, we should consider measurement classes like $\cM_s$ in Fig.~\ref{fig:classes} 
that are closer to the class of POVMs ($\cM_0$ in Fig.~\ref{fig:classes}) than the measurement class of SEP. 
If a measurement class is sufficiently close to the class of POVMs, 
it is hard for an experiment to deny the model 
because the difference between the experiment and model might be due to an experimental error.
In this paper, in order to deny such an alternative measurement class theoretically,
we investigate whether an extended measurement class drastically improves perfect discrimination of separable states 
even when it sufficiently approximates the class of POVMs.
For this aim, we define slightly larger measurement classes than that of POVMs to satisfy the following three conditions, 
and investigate what happens in adopting them.
(i) The measurement classes contain the class of POVMs. 
(ii) Non-negativity of probability holds for every separable state.
(iii) The measurement classes are represented as a continuous one-parameter family with the parameter $s$, 
and the case $s=0$ is just the class of POVMs.
For small $s>0$, the measurement class can be regarded as an approximation of the class of POVMs.

We consider two types of measurement classes.
The first one $\cM_s$ is given by the restriction of negative eigenvalues of measurement elements 
while eigenvalues of POVM elements are restricted to be non-negative.
The second one $\cM(\cK_s)$ is defined by the positive cone that is given as the sum of positive semi-definite matrices
and positive partial transpose of restricted entangled vectors.
As a result, we show that 
the performance of perfect discrimination is dramatically improved unless $s=0$, 
which implies that {}the above approximate classes of POVMs are unlikely to exist.

The remainder of this paper is organized as follows.
Section~\ref{sect2} describes QT in the framework of GPTs.
Section~\ref{sect3} defines approximate QT on a bipartite system.
Section~\ref{sect4} gives our main results and 
shows a drastic improvement of multiple-copy state discrimination.
Section~\ref{sect5} is a conclusion.

\section{Framework of GPTs}\label{sect2}
Throughout this paper, we only consider finite-dimensional systems.
As already stated, there are state/measurement classes in a GPT.
First, let us describe a state class.
To handle both pure and mixed states, 
the set of all states must be a closed convex set.
In GPTs, the set $\cS(\cK,u)$ of all states is defined 
by the intersection of an affine hyperplane and a positive cone: 
\[
\cS(\cK,u) = \set{x\in\cK | \langle x,u\rangle=1},
\]
where $\cK$ is a \textit{positive cone} of a finite-dimensional real Hilbert space $\cV$ 
equipped with an inner product $\langle\cdot,\cdot\rangle$, 
and the \textit{unit effect} $u$ is an interior point of the \textit{dual cone} $\cK^\ast$.
Here, $\cK$ is called a \textit{positive cone} 
if $\cK$ is a closed convex set satisfying that 
\begin{itemize}
	\item
	$\alpha x\in\cK$ for all $\alpha\ge0$ and $x\in\cK$;
	\item
	$\cK$ has non-empty interior;
	\item
	$\cK\cap(-\cK)=\{0\}$.
\end{itemize}
Also, for a positive cone $\cK$, 
the \textit{dual cone} $\cK^\ast$ is defined as 
\[
\cK^\ast = \set{y\in\mathcal{V} | \forall x\in\cK,\ \langle x,y\rangle\ge0},
\]
which is also a positive cone.

Next, let us describe a measurement class $\cM$ 
(of a GPT with a state class $\cS(\cK,u)$) 
by using non-negativity of probability.
A measurement is given as a family $\{y_i\}_{i=1}^n$, 
where $\{1,2,\ldots,n\}$ denotes the set of outcomes.
If a state $x\in\cS(\cK,u)$ is measured by a measurement $\{y_i\}_{i=1}^n$, 
then each outcome $i$ is obtained with probability $\langle x,y_i\rangle$.
Therefore, we need the following postulate.

\begin{postulate}[Non-negativity of probability]\label{post1}
	For each state $x\in\cS(\cK,u)$ and each measurement $\{y_i\}_{i=1}^n\in\cM$, 
	the family $\{\langle x,y_i\rangle\}_{i=1}^n$ is a probability distribution, i.e., 
	$\langle x,y_i\rangle\ge0$ for all outcomes $i$, 
	and $\sum_{j=1}^n \langle x,y_j\rangle=1$.
\end{postulate}

Due to Postulate~\ref{post1}, 
each measurement element $y_i$ must lie in $\cK^\ast$.
The largest measurement class with Postulate~\ref{post1} is given as 
\[
\cM(\cK^\ast,u) = \Set{\{y_i\}_{i=1}^n | 
\begin{array}{c}
	n\in\mathbb{N},\ \sum_{j=1}^n y_j=u,\\
	y_i\in\cK^\ast\ (\forall i)
\end{array}},
\]
but we do not assume that $\cM$ is the largest one.

Now, let us describe the state/measurement classes of QT 
by using the above framework of GPTs.
Assume that 
\begin{enumerate}
	\renewcommand{\theenumi}{QT\arabic{enumi}}
	\setlength{\leftskip}{2eM}
	\item
	$\cV$ is the set $\cT(\cH)$ of all Hermitian matrices 
	on a finite-dimensional complex Hilbert space $\cH$; 
	\item
	An inner product on $\cT(\cH)$ is given by $\langle X,Y\rangle=\Tr XY$; 
	\item
	$\cK$ is the set $\cT_+(\cH)$ of all positive semi-definite matrices on $\cH$; 
	\item
	$u$ is the identity matrix $I$ on $\cH$; 
	\item
	$\cM=\cM(\cT_+(\cH),I)$.
\end{enumerate}
Then $\cS(\cT_+(\cH),I)$ equals the set of all density matrices on $\cH$, 
and $\cM$ equals the class of POVMs.
Since these classes are usual ones in QT, 
it turns out that QT is a typical example of GPTs.


\textit{Perfect discrimination.}---%
Let $\{x_i\}_{i=1}^n$ be a family of $n$ states in $\cS(\cK,u)$. 
We say that $\{x_i\}_{i=1}^n$ is \textit{perfectly distinguishable} 
if there exists a measurement $\{y_j\}_{j=1}^n\in\cM$ 
such that $\braket{x_i,y_j}=\delta_{ij}$, 
where $\delta_{ij}$ denotes the Kronecker delta. 
In this paper, we address the case $n=2$ mainly.

\section{Approximate QT}\label{sect3}
We consider a bipartite system of two finite-dimensional quantum systems 
$\cH_A$ (Alice's system) and $\cH_B$ (Bob's system), 
but the bipartite system is not necessarily QT.
More precisely, we assume QT1, QT2, and QT4 for $\cH=\cH_A\otimes\cH_B$, 
but do not necessarily assume QT3 or QT5.
Let us consider such a bipartite system in the framework of GPTs.
When Alice and Bob prepare quantum states $\rho^A$ and $\rho^B$ independently, 
the product states $\rho^A\otimes\rho^B$ is prepared on the bipartite system.
Considering the convexity of a state class, 
we need the following postulate.

\begin{postulate}\label{post2}
	A state class of the bipartite system contains all separable states.
\end{postulate}

Hereinafter, $\cT(\cH_A)$ is denoted by $\cT(A)$.
The notations $\cT_+(A)$, $\cT(B)$, $\cT_+(B)$, 
$\cT(AB)$, and $\cT_+(AB)$ are similarly defined.
Also, since the unit effect of the bipartite system is always $I$, 
we denote $\cM(\cK,I)$ by $\cM(\cK)$ simply.
Let $\mathrm{SEP}(A;B)$ be the set 
\[
\Set{\sum_{i=1}^n X_i^A\otimes X_i^B | 
\begin{array}{c}
	n\in\mathbb{N},\ X_i^A\in\cT_+(A),\\
	X_i^B\in\cT_+(B)\ (\forall i)
\end{array}}.
\]
Ref.\ \cite{Arai2019} used the largest measurement class $\cM(\mathrm{SEP}(A;B)^\ast)$ 
to discriminate two separable pure states, 
but their measurement class is too far from the class of POVMs.
Therefore, we need to define a measurement class 
that is sufficiently close to the class of POVMs.
Moreover, for some state class with Postulate~\ref{post2}, 
the measurement class must satisfy Postulate~\ref{post1}.
Let us define such measurement classes in two different ways here.

\begin{definition}[Measurement class]
	For $s\ge0$, we define the measurement class $\cM_s$ as 
	\begin{equation*}
		\cM_s = \Set{\{M_i\}_{i=1}^n | 
		\begin{array}{c}
			n\in\mathbb{N},\ \sum_{j=1}^n M_j=I,\\
			M_i\in\mathrm{SEP}(A;B)^\ast,\\
			\nege(M_i)\le s\ (\forall i)
		\end{array}},
	\end{equation*}
	where for $X\in\cT(AB)$ the value $\nege(X)$ is defined as 
	\[
	\nege(X) =
	\begin{dcases}
		\max_{\substack{\lambda<0\\ \text{ eigenvalue of }X}} \abs{\lambda}
		&\text{if $X$ has a negative eigenvalue},\\
		0 & \text{otherwise}.
	\end{dcases}
	\]
\end{definition}

To define another one-parameter family $\cM(\cK_s)$ of measurement classes, 
we define the following special positive cones.

\begin{definition}[One-parameter family of positive cones]
	For a vector $v\in\cH_A\otimes\cH_B$, let $\sco(v)$ be the value 
	\[
	\sco(v) = 
	\begin{cases}
		\lambda_1\lambda_2 & v\not=0,\\
		0 & v=0,
	\end{cases}
	\]
	where $\lambda_1 \ge \lambda_2 \ge \cdots \ge \lambda_d$, $d=\min\{\dim\cH_A,\dim\cH_B\}$, 
	denote the Schmidt coefficients of $v/\|v\|$.
	Then, for $s\in[0,1/2]$, we define the positive cones $\cK_s^{(0)}$ and $\cK_s$ as 
	\begin{align*}
		\cK_s^{(0)} &= \conv\{ \ketbra{v}{v} \mid v\in\cH_A\otimes\cH_B,\ \sco(v)\le s \},\\
		\cK_s &= \cT_+(AB) + \Gamma(\cK_s^{(0)}),
	\end{align*}
	where $\conv(\cX)$ denotes the convex hull of a set $\cX\subset\cT(AB)$ 
	and $\Gamma$ denotes the partial transpose on Bob's system,
	i.e., $\Gamma$ is the linear map defined by the tensor product of identity map and transposition.
\end{definition}

The value $\sco(v)$ is closely related to negative eigenvalues of $\Gamma(\ketbra{v}{v})$: 
for every $v\in\cH_A\otimes\cH_B$ 
\begin{equation}
	\|v\|^2\sco(v) = \nege(\Gamma(\ketbra{v}{v})).
	\label{sc-neg}
\end{equation}
Eq.~\eqref{sc-neg} follows from the fact that, if $v\not=0$, 
the set of all eigenvalues of $\Gamma(\ketbra{v}{v})/\|v\|^2$ is 
$\set{\pm\lambda_i\lambda_j,\,\lambda_k^2 | 1\le i<j\le d,\ 1\le k\le d}$, 
where $\lambda_1\ge\cdots\ge\lambda_d$ denote the Schmidt coefficients of $v/\|v\|$.
Hence, the inequality $\sco(v)\le s$ is a restriction of negative eigenvalues of elements of $\cK_s$.
Since the Schmidt coefficients of a unit vector $v\in\cH_A\otimes\cH_B$ 
represent the amount of entanglement about the pure state $\ketbra{v}{v}$, 
one can also regard the inequality $\sco(v)\le s$ as a restriction of entanglement 
on the inside of the partial transpose $\Gamma$.
Also, once the parameter $s$ increases, 
the positive cones $\cK_s^{(0)}$ and $\cK_s$ become larger.
Thus, the following inclusion relations hold: 
\begin{gather*}
	\mathrm{SEP}(A;B) = \cK_0^{(0)} \subset \cK_s^{(0)} \subset \cK_{1/2}^{(0)} = \cT_+(AB),\\
	\cT_+(AB) = \cK_0 \subset \cK_s \subset \cK_{1/2} \subset \mathrm{SEP}(A;B)^\ast.
\end{gather*}
Note that the classes $\cM_0$ and $\cM(\cK_0)$ are the class of POVMs. 
Also, since $\cK_s^{(0)}$ satisfies \textit{local unitary invariance}, 
i.e., $(U_A\otimes U_B)\cK_s^{(0)}(U_A\otimes U_B)^\dag = \cK_s^{(0)}$ for all unitary matrices $U_A$ and $U_B$, 
no positive cones $\cK_s$ depend on an orthonormal basis of $\cH_B$ that defines the partial transpose $\Gamma$.

\section{Perfect discrimination in approximate QT}\label{sect4}
Let us consider perfect discrimination of separable pure states 
by measurements of $\cM_s$ and $\cM(\cK_s)$.
First, for separable pure states that are parameterized, 
we give concrete measurements in the case $\dim\cH_A=\dim\cH_B=2$.
Let $\rho_1$ and $\rho_2$ be the separable pure states given as 
\begin{equation}
	\rho_1 = 
	\begin{bmatrix}
		1&0\\
		0&0
	\end{bmatrix}
	\otimes
	\begin{bmatrix}
		1&0\\
		0&0
	\end{bmatrix}
	,\quad \rho_2 = 
	\begin{bmatrix}
		1-\alpha_1&\beta_1\\
		\beta_1&\alpha_1
	\end{bmatrix}
	\otimes
	\begin{bmatrix}
		1-\alpha_2&\beta_2\\
		\beta_2&\alpha_2
	\end{bmatrix}, \label{states}
\end{equation}
where $\alpha_1,\alpha_2\in[0,1]$ and $\beta_i = \sqrt{\alpha_i(1-\alpha_i)}$.
If the relations $s\in[0,1/2]$ and 
\begin{equation}
	(1-\alpha_1)(1-\alpha_2) \le 4s^2\alpha_1\alpha_2 \label{S1}
\end{equation}
hold, then $\rho_1$ and $\rho_2$ are perfectly distinguishable 
by some measurement $\{T_i+\Gamma(T_i)\}_{i=1,2}\in\cM_s$.
The measurement $\{T_i+\Gamma(T_i)\}_{i=1,2}\in\cM_s$ is given below 
except for the trivial cases $\alpha_1=1$ and $\alpha_2=1$:
If $\gamma \coloneqq \alpha_1+\alpha_2>1$, then 
\begin{gather*}
	2\gamma T_1 = \gamma\ketbra{v_1}{v_1} + (\gamma-1)\ketbra{v_2}{v_2} + (\gamma-1)\ketbra{v_3}{v_3},\label{eq01}\\
	v_1 =
	\begin{bmatrix}
		1\\
		0
	\end{bmatrix}
	\otimes
	\begin{bmatrix}
		1\\
		0
	\end{bmatrix}
	- \frac{\beta_1\beta_2}{\alpha_1\alpha_2}
	\begin{bmatrix}
		0\\
		1
	\end{bmatrix}
	\otimes
	\begin{bmatrix}
		0\\
		1
	\end{bmatrix}
	,\label{eq01'}\\
	v_2 =
	\begin{bmatrix}
		1\\
		-\beta_1/\alpha_1
	\end{bmatrix}
	\otimes
	\begin{bmatrix}
		0\\
		1\\
	\end{bmatrix}
	,\quad
	v_3 =
	\begin{bmatrix}
		0\\
		1\\
	\end{bmatrix}
	\otimes
	\begin{bmatrix}
		1\\
		-\beta_2/\alpha_2
	\end{bmatrix}
	,\label{eq01''}\\
	T_2=(U_A\otimes U_B)T_1(U_A\otimes U_B)^\dag, \label{local-uni}\\
	U_A = \frac{1}{\sqrt{\alpha_1}}
	\begin{bmatrix}
		\beta_1&\alpha_1\\
		\alpha_1&-\beta_1
	\end{bmatrix}
	,\quad
	U_B = \frac{1}{\sqrt{\alpha_2}}
	\begin{bmatrix}
		\beta_2&\alpha_2\\
		\alpha_2&-\beta_2
	\end{bmatrix};
\end{gather*}
if $\gamma=1$, then 
\begin{equation*}
	T_1 = \frac{1}{2}
	\begin{bmatrix}
		1 & 0 & 0 & -1\\
		0 & 0 & 0 & 0\\
		0 & 0 & 0 & 0\\
		-1 & 0 & 0 & 1
	\end{bmatrix}
	,\quad
	T_2 = \frac{1}{2}
	\begin{bmatrix}
		0 & 0 & 0 & 0\\
		0 & 1 & 1 & 0\\
		0 & 1 & 1 & 0\\
		0 & 0 & 0 & 0
	\end{bmatrix}.
\end{equation*}
When a measurement class is $\cM(\cK_s)$, 
Eq.~\eqref{S1} turns to 
\begin{equation}
	(1-\alpha_1)(1-\alpha_2) \le t\alpha_1\alpha_2, \label{S2}
\end{equation}
where $s=\sqrt{t}/(1+t)$ and $t\in[0,1]$.

A simple calculation ensures that the above measurements indeed discriminate the states \eqref{states} perfectly.
Thus, we only have to examine whether the above measurements are contained in $\cM_s$ and $\cM(\cK_s)$.
For details, see supplemental material.

Next, let us consider the general case $\dim\cH_A,\dim\cH_B\ge2$.
Let $\rho_1=\rho_1^A\otimes\rho_1^B$ and $\rho_2=\rho_2^A\otimes\rho_2^B$ be separable pure states.
We can take orthonormal bases of $\cH_A$ and $\cH_B$ such that 
$\rho_1$ and $\rho_2$ are expressed as \eqref{states}, 
i.e., their representation matrices are given by the direct sums of the matrices \eqref{states} and the zero matrix.
Therefore, the general case is reduced to the case $\dim\cH_A=\dim\cH_B=2$, 
and we obtain the following theorems.

\begin{theorem}[Perfect discrimination with $\cM_s$]\label{thm1}
	If $x \coloneqq \Tr\rho_1^A\rho_2^A$ and $y \coloneqq \Tr\rho_1^B\rho_2^B$ satisfy 
	the relations $s\in[0,1/2]$ and 
	\begin{equation}
		xy \le 4s^2(1-x)(1-y), \label{S1'}
	\end{equation}
	then $\rho_1$ and $\rho_2$ are perfectly distinguishable 
	by some measurement of $\cM_s$.
\end{theorem}

\begin{theorem}[Perfect discrimination with $\cM(\cK_s)$]\label{thm2}
	If $x \coloneqq \Tr\rho_1^A\rho_2^A$ and $y \coloneqq \Tr\rho_1^B\rho_2^B$ satisfy 
	the relations $s=\sqrt{t}/(1+t)$, $t\in[0,1]$, and 
	\begin{equation}
		xy \le t(1-x)(1-y), \label{S2'}
	\end{equation}
	then $\rho_1$ and $\rho_2$ are perfectly distinguishable 
	by some measurement of $\cM(\cK_s)$.
\end{theorem}

Using Theorems~\ref{thm1} and \ref{thm2}, 
we find the following drastic improvement of multiple-copy state discrimination.
Let $\sigma_1$ and $\sigma_2$ be distinct pure states on a single quantum system.
In QT, the non-trivial $n$-copies $\sigma_1^{\otimes n}$ and $\sigma_2^{\otimes n}$ 
never be perfectly distinguishable, 
where we say that the $n$-copies $\sigma_1^{\otimes n}$ and $\sigma_2^{\otimes n}$ are \textit{non-trivial} 
if $\sigma_1$ and $\sigma_2$ are distinct and non-orthogonal.
However, it is known \cite{Arai2019} that, for some finite $n$, 
the non-trivial $n$-copies $\sigma_1^{\otimes n}$ and $\sigma_2^{\otimes n}$ 
are perfectly distinguishable by some measurement of $\cM(\mathrm{SEP}^\ast(A;B))$.
Surprisingly, the same statement is true 
for the measurement classes $\cM_s$ and $\cM(\cK_s)$ 
that are slightly larger than the class of POVMs.
To see this fact, regarding the $2n$-copies $\sigma_1^{\otimes2n} = \sigma_1^{\otimes n}\otimes\sigma_1^{\otimes n}$ 
and $\sigma_2^{\otimes2n} = \sigma_2^{\otimes n}\otimes\sigma_2^{\otimes n}$ as bipartite separable pure states, 
we apply Theorem~\ref{thm1} to them.
Assume $s\in(0,1/2]$. Since 
\[
x = y = \Tr\sigma_1^{\otimes n}\sigma_2^{\otimes n} = (\Tr\sigma_1\sigma_2)^n \xrightarrow{n\to\infty} 0,
\]
a sufficiently large $n$ satisfies the inequality \eqref{S1'}.
Thus, for some finite $n$, the $2n$-copies $\sigma_1^{\otimes2n}$ and $\sigma_2^{\otimes2n}$ 
are perfectly distinguishable by some measurement of $\cM_s$.
Also, since a sufficiently large $n$ satisfies the inequality \eqref{S2'}, 
the same statement is true for $\cM(\cK_s)$.
We summarize these facts as the following corollary and Table~\ref{table}.

\begin{corollary}[Multiple-copy state discrimination]\label{coro1}
	Assume $s\in(0,1/2]$. 
	Then, for some finite $n\in\mathbb{N}$, 
	the $2n$-copies $\sigma_1^{\otimes2n} = \sigma_1^{\otimes n}\otimes\sigma_1^{\otimes n}$ 
	and $\sigma_2^{\otimes2n} = \sigma_2^{\otimes n}\otimes\sigma_2^{\otimes n}$ 
	are perfectly distinguishable by some measurement of $\cM_s$.
	The same statement is true for $\cM(\cK_s)$.
\end{corollary}

\begin{table*}
	\centering
	\caption{Finite-copy perfect discrimination for each measurement class.
	Assume $s\in(0,1/2]$ here.}\label{table}
	\vspace{1ex}
	\begin{tabular}{ccccc}
		\hline
		Measurement class&$\cM(\cT_+(AB))$&$\cM_s$&$\cM(\cK_s)$&$\cM(\mathrm{SEP}(A;B)^\ast)$\\
		\hline
		Perfect discrimination&\multirow{2}{*}{Impossible}&Possible&Possible&Possible\\
		of non-trivial $n$-copies&   &for finite $n$&for finite $n$&for finite $n$\\
		\hline
	\end{tabular}
\end{table*}

The domains of $(x,y)$ in Theorems~\ref{thm1} and \ref{thm2} are illustrated as Fig.~\ref{fig:xy}.
Once the parameter $s\in[0,1/2]$ decreases, 
the domains of $(x,y)$ in Theorems~\ref{thm1} and \ref{thm2} becomes smaller.
However, the origin is an interior point of the domain of $(x,y)$ 
(as a subspace of the square $[0,1]^2$) unless $s$ is zero.
This fact is important to understand Corollary~\ref{coro1}.
To see this importance, recall that 
the value $x_n=y_n=(\Tr\sigma_1\sigma_2)^n$ converges to zero as $n\to\infty$.
As long as the origin is an interior point of the domain of $(x,y)$, 
for some finite $n$ the point $(x_n,y_n)$ lies in the domain of $(x,y)$.
Since the origin is an interior point of the domains of $(x,y)$ in Theorems~\ref{thm1} and \ref{thm2}, 
we can check Corollary~\ref{coro1} again.

\begin{figure}
	\centering
	\includegraphics[scale=0.3]{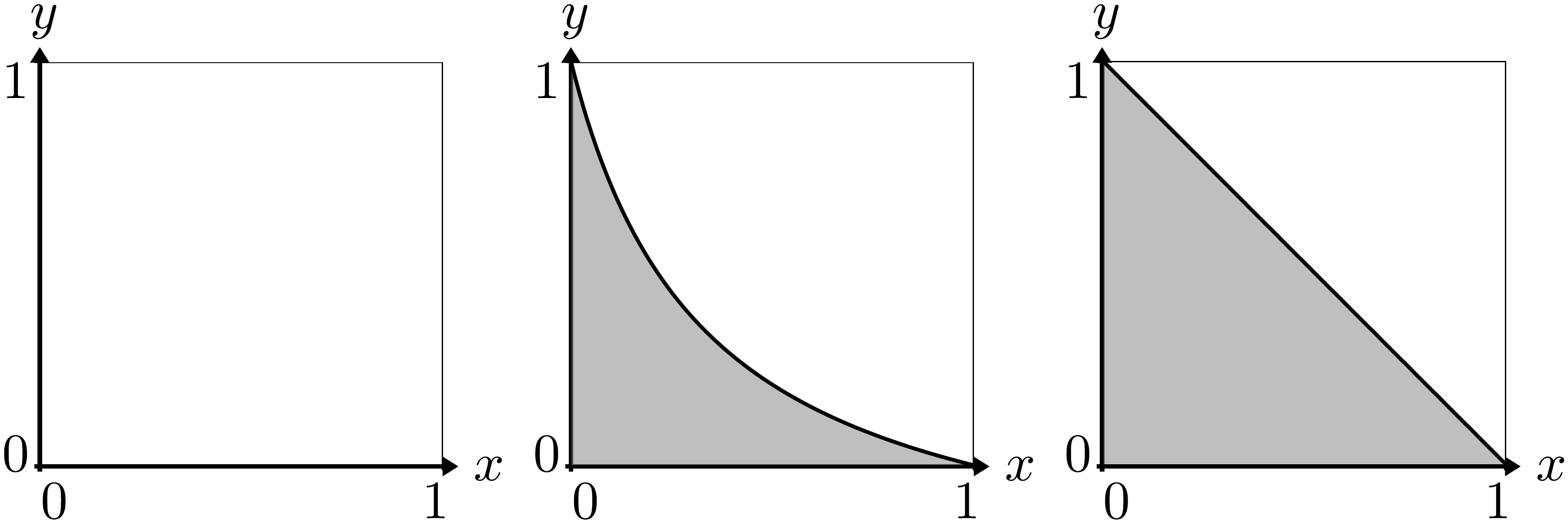}
	\caption{The domains of $(x,y)$ in Theorem~\ref{thm1} (resp.\ Theorem~\ref{thm2}) 
	for $s=0,1/4,1/2$ (resp.\ $(s,t)=(0,0),(2/5,1/4),(1/2,1)$).}
	\label{fig:xy}
\end{figure}

\section{Conclusion}\label{sect5}
To investigate the performance of perfect discrimination in models approximately close to QT, 
we have defined the two measurement classes $\cM_s$ and $\cM(\cK_s)$ 
that are smoothly connected with the class of POVMs ($s=0$).
As a result, unless $s=0$, 
the performance of perfect discrimination is drastically improved 
for both the measurement classes $\cM_s$ and $\cM(\cK_s)$.
More precisely, their measurements enable us 
to discriminate non-trivial $2n$-copies perfectly for some finite $n$.
This result suggests that the approximate measurement classes $\cM_s$ and $\cM(\cK_s)$ are unlikely to exist.

Although we have shown perfect discrimination of non-trivial $2n$-copies for some finite $n$, 
it is interesting to investigate the converse: 
\textit{Let $\cM$ be a measurement class.
If there never exist $n\in\mathbb{N}$ and non-trivial $n$-copies such that 
the $n$-copies are perfectly distinguishable by some measurement of $\cM$, 
then is $\cM$ contained in POVMs?} It is a future work.
Also, to consider another problem, 
assume that there exist pure states $\sigma_1$ and $\sigma_2$ such that 
for every $n\in\mathbb{N}$ the $n$-copies $\sigma_1^{\otimes n}$ and $\sigma_2^{\otimes n}$ 
are not perfectly distinguishable by any measurements of $\cM$.
Then it is also interesting to examine the error probability 
in discriminating $\sigma_1^{\otimes n}$ and $\sigma_2^{\otimes n}$.
If $\cM$ is the class of POVMs, 
the error probability is exponentially decreasing 
and the exponential decreasing rate is known \cite[Section~3]{HayashiBook:2017}.
However, we are interested in the case where $\cM$ is general.
It is another future work.

\ack
YY was supported by Japan Society for the Promotion of Science (JSPS) 
Grant-in-Aid for JSPS Fellows No.\ 19J20161. 
MH was supported in part by JSPS Grant-in-Aids for 
Scientific Research (A) No.\ 17H01280 and for Scientific Research (B) No.\ 16KT0017, 
and Kayamori Foundation of Information Science Advancement.

\section*{Appendix: Proofs of technical lemmas}
In this appendix, we prove technical lemmas, 
which yield Theorems~\ref{thm1} and \ref{thm2}.

\begin{lemma}[Perfect discrimination with $\cM_s$]\label{supp-mat-lem1}
	Let $\rho_1$ and $\rho_2$ be the separable pure states given as 
	\begin{equation}
		\rho_1 = 
		\begin{bmatrix}
			1&0\\
			0&0
		\end{bmatrix}
		\otimes
		\begin{bmatrix}
			1&0\\
			0&0
		\end{bmatrix}
		,\quad \rho_2 = 
		\begin{bmatrix}
			1-\alpha_1&\beta_1\\
			\beta_1&\alpha_1
		\end{bmatrix}
		\otimes
		\begin{bmatrix}
			1-\alpha_2&\beta_2\\
			\beta_2&\alpha_2
		\end{bmatrix}, \label{supp-mat-states}
	\end{equation}
	where $\alpha_1,\alpha_2\in[0,1]$ and $\beta_i = \sqrt{\alpha_i(1-\alpha_i)}$.
	If the relations $s\in[0,1/2]$ and 
	\begin{equation}
		(1-\alpha_1)(1-\alpha_2) \le 4s^2\alpha_1\alpha_2 \label{supp-mat-S1}
	\end{equation}
	hold, then $\rho_1$ and $\rho_2$ are perfectly distinguishable 
	by some measurement $\{T_i+\Gamma(T_i)\}_{i=1,2}\in\cM_s$.
	The measurement $\{T_i+\Gamma(T_i)\}_{i=1,2}\in\cM_s$ is given below 
	except for the trivial cases $\alpha_1=1$ and $\alpha_2=1$:
	If $\gamma \coloneqq \alpha_1+\alpha_2>1$, then 
	\begin{gather}
		2\gamma T_1 = \gamma\ketbra{v_1}{v_1} + (\gamma-1)\ketbra{v_2}{v_2} + (\gamma-1)\ketbra{v_3}{v_3},\label{supp-mat-eq01}\\
		v_1 =
		\begin{bmatrix}
			1\\
			0
		\end{bmatrix}
		\otimes
		\begin{bmatrix}
			1\\
			0
		\end{bmatrix}
		- \frac{\beta_1\beta_2}{\alpha_1\alpha_2}
		\begin{bmatrix}
			0\\
			1
		\end{bmatrix}
		\otimes
		\begin{bmatrix}
			0\\
			1
		\end{bmatrix}
		,\label{supp-mat-eq01'}\\
		v_2 =
		\begin{bmatrix}
			1\\
			-\beta_1/\alpha_1
		\end{bmatrix}
		\otimes
		\begin{bmatrix}
			0\\
			1\\
		\end{bmatrix}
		,\quad
		v_3 =
		\begin{bmatrix}
			0\\
			1\\
		\end{bmatrix}
		\otimes
		\begin{bmatrix}
			1\\
			-\beta_2/\alpha_2
		\end{bmatrix}
		,\label{supp-mat-eq01''}\\
		T_2=(U_A\otimes U_B)T_1(U_A\otimes U_B)^\dag, \label{supp-mat-local-uni}\\
		U_A = \frac{1}{\sqrt{\alpha_1}}
		\begin{bmatrix}
			\beta_1&\alpha_1\\
			\alpha_1&-\beta_1
		\end{bmatrix}
		,\quad
		U_B = \frac{1}{\sqrt{\alpha_2}}
		\begin{bmatrix}
			\beta_2&\alpha_2\\
			\alpha_2&-\beta_2
		\end{bmatrix};
	\end{gather}
	if $\gamma=1$, then 
	\begin{equation*}
		T_1 = \frac{1}{2}
		\begin{bmatrix}
			1 & 0 & 0 & -1\\
			0 & 0 & 0 & 0\\
			0 & 0 & 0 & 0\\
			-1 & 0 & 0 & 1
		\end{bmatrix}
		,\quad
		T_2 = \frac{1}{2}
		\begin{bmatrix}
			0 & 0 & 0 & 0\\
			0 & 1 & 1 & 0\\
			0 & 1 & 1 & 0\\
			0 & 0 & 0 & 0
		\end{bmatrix}.
	\end{equation*}
\end{lemma}

\begin{lemma}[Perfect discrimination with $\cM(\cK_s)$]\label{supp-mat-lem2}
	Let $\rho_1$ and $\rho_2$ be the separable pure states given as \eqref{supp-mat-states}.
	If the relations $s=\sqrt{t}/(1+t)$, $t\in[0,1]$, and 
	\begin{equation}
		(1-\alpha_1)(1-\alpha_2) \le t\alpha_1\alpha_2, \label{supp-mat-S2}
	\end{equation}
	hold, then $\rho_1$ and $\rho_2$ are perfectly distinguishable 
	by the measurement $\{T_i+\Gamma(T_i)\}_{i=1,2}\in\cM_s$ given in Lemma~$\ref{supp-mat-lem1}$.
\end{lemma}

We prove Lemmas~\ref{supp-mat-lem2} and \ref{supp-mat-lem1} in this order.

\begin{proof}[Proof of Lemma~$\ref{supp-mat-lem2}$]
	Assume that $s=\sqrt{t}/(1+t)$, $t\in[0,1]$, and \eqref{supp-mat-S2}.
	All we need is to show that 
	\begin{enumerate}
		\item
		$T_1+T_2 + \Gamma(T_1+T_2) = I$,
		\item
		$T_i\in\cK_s^{(0)}$ for all $i=1,2$,
		\item
		$\Tr\rho_1T_2 = \Tr\rho_2T_1 = 0$.
	\end{enumerate}
	Indeed, if (i) and (ii) hold, then $\{T_i+\Gamma(T_i)\}_{i=1,2}\in\cM(\cK_s)$.
	Also, if (i) and (iii) hold, 
	then the equations $\Gamma(\rho_i)=\rho_i$, $i=1,2$, imply that 
	$\Tr\rho_i(T_j+\Gamma(T_j)) = 2\Tr\rho_iT_j = \delta_{ij}$ for all $i,j\in\{1,2\}$.
	Therefore, if (i)--(iii) hold, then Lemma~\ref{supp-mat-lem1} follows.
	Also, note that  
	$(1-\alpha_2)(1-\alpha_1) \le t\alpha_1\alpha_2 \le \alpha_1\alpha_2$ 
	thanks to $t\in[0,1]$ and \eqref{supp-mat-S2}.
	Thus $\gamma=\alpha_1+\alpha_2\ge1$.
	If $\alpha_1\alpha_2=0$, then $\alpha_1=1$ or $\alpha_2=1$, 
	which is a trivial case.
	Therefore, without loss of generality, we may assume $\alpha_1\alpha_2>0$.
	\par
	Proof of (i). First, assume $\gamma=1$.
	Then 
	\[
	T_1+T_2 + \Gamma(T_1+T_2) = \frac{1}{2}
	\begin{bmatrix}
		1&0&0&-1\\
		0&1&1&0\\
		0&1&1&0\\
		-1&0&0&1\\
	\end{bmatrix}
	+ \frac{1}{2}
	\begin{bmatrix}
		1&0&0&1\\
		0&1&-1&0\\
		0&-1&1&0\\
		1&0&0&1\\
	\end{bmatrix}
	= I.
	\]
	Next, assume $\gamma>1$. Put $w_i=(U_A\otimes U_B)v_i$ for $i=1,2,3$.
	Then $w_i$, $i=1,2,3$, can be calculated as follows: 
	\begin{align*}
		w_1 &= \frac{1}{\sqrt{\alpha_1\alpha_2}}\biggl(
		\begin{bmatrix}
			\beta_1\\
			\alpha_1
		\end{bmatrix}
		\otimes
		\begin{bmatrix}
			\beta_2\\
			\alpha_2
		\end{bmatrix}
		- \frac{\beta_1\beta_2}{\alpha_1\alpha_2}
		\begin{bmatrix}
			\alpha_1\\
			-\beta_1
		\end{bmatrix}
		\otimes
		\begin{bmatrix}
			\alpha_2\\
			-\beta_2
		\end{bmatrix}
		\biggr)\\
		&= 
		\frac{1}{\sqrt{\alpha_1\alpha_2}}\Biggl(
		\begin{bmatrix}
			\beta_1\beta_2\\
			\beta_1\alpha_2\\
			\alpha_1\beta_2\\
			\alpha_1\alpha_2
		\end{bmatrix}
		- \frac{\beta_1\beta_2}{\alpha_1\alpha_2}
		\begin{bmatrix}
			\alpha_1\alpha_2\\
			-\alpha_1\beta_2\\
			-\beta_1\alpha_2\\
			\beta_1\beta_2
		\end{bmatrix}
		\Biggr)
		= \frac{1}{\sqrt{\alpha_1\alpha_2}}
		\begin{bmatrix}
			0\\
			\beta_1\\
			\beta_2\\
			\gamma-1
		\end{bmatrix},
	\end{align*}
	\begin{equation*}
		w_2 = \frac{1}{\sqrt{\alpha_1\alpha_2}}
		\begin{bmatrix}
			0\\
			1
		\end{bmatrix}
		\otimes
		\begin{bmatrix}
			\alpha_2\\
			-\beta_2
		\end{bmatrix}
		= \sqrt{\frac{\alpha_2}{\alpha_1}}v_3,\quad
		w_3 = \frac{1}{\sqrt{\alpha_1\alpha_2}}
		\begin{bmatrix}
			\alpha_1\\
			-\beta_1
		\end{bmatrix}
		\otimes
		\begin{bmatrix}
			0\\
			1
		\end{bmatrix}
		= \sqrt{\frac{\alpha_1}{\alpha_2}}v_2.
	\end{equation*}
	Thus, putting $\xi=\beta_1\beta_2/\alpha_1\alpha_2$, we have 
	\begin{align*}
		&\quad T_1+T_2 = \frac{1}{2}(\ketbra{v_1}{v_1} + \ketbra{w_1}{w_1})
		+ \frac{\gamma-1}{2\gamma}(\ketbra{v_2}{v_2} + \ketbra{v_3}{v_3} + \ketbra{w_2}{w_2} + \ketbra{w_3}{w_3})\\
		&= \frac{1}{2}(\ketbra{v_1}{v_1} + \ketbra{w_1}{w_1})
		+ \frac{\gamma-1}{2\gamma}\Bigl( \ketbra{v_2}{v_2} + \ketbra{v_3}{v_3}
		+ \frac{\alpha_2}{\alpha_1}\ketbra{v_3}{v_3} + \frac{\alpha_1}{\alpha_2}\ketbra{v_2}{v_2} \Bigr)\\
		&= \frac{1}{2}(\ketbra{v_1}{v_1} + \ketbra{w_1}{w_1})
		+ \frac{\gamma-1}{2}\Bigl( \frac{1}{\alpha_2}\ketbra{v_2}{v_2} + \frac{1}{\alpha_1}\ketbra{v_3}{v_3} \Bigr)\\
		&= \frac{1}{2}
		\begin{bmatrix}
			1&0&0&-\xi\\
			0&0&0&0\\
			0&0&0&0\\
			-\xi&0&0&\xi^2
		\end{bmatrix}
		+ \frac{1}{2\alpha_1\alpha_2}
		\begin{bmatrix}
			0&0&0&0\\
			0&\beta_1^2&\beta_1\beta_2&(\gamma-1)\beta_1\\
			0&\beta_1\beta_2&\beta_2^2&(\gamma-1)\beta_2\\
			0&(\gamma-1)\beta_1&(\gamma-1)\beta_2&(\gamma-1)^2
		\end{bmatrix}
		\\
		&\quad+ \frac{\gamma-1}{2\alpha_1\alpha_2}\biggl(
		\begin{bmatrix}
			\alpha_1&-\beta_1\\
			-\beta_1&1-\alpha_1
		\end{bmatrix}
		\otimes
		\begin{bmatrix}
			0&0\\
			0&1
		\end{bmatrix}
		+
		\begin{bmatrix}
			0&0\\
			0&1
		\end{bmatrix}
		\otimes
		\begin{bmatrix}
			\alpha_2&-\beta_2\\
			-\beta_2&1-\alpha_2
		\end{bmatrix}
		\biggr)\\
		&= \frac{1}{2}
		\begin{bmatrix}
			1&0&0&-\xi\\
			0&0&0&0\\
			0&0&0&0\\
			-\xi&0&0&\xi^2
		\end{bmatrix}
		+ \frac{1}{2\alpha_1\alpha_2}
		\begin{bmatrix}
			0&0&0&0\\
			0&\beta_1^2&\beta_1\beta_2&(\gamma-1)\beta_1\\
			0&\beta_1\beta_2&\beta_2^2&(\gamma-1)\beta_2\\
			0&(\gamma-1)\beta_1&(\gamma-1)\beta_2&(\gamma-1)^2
		\end{bmatrix}
		\\
		&\quad+ \frac{\gamma-1}{2\alpha_1\alpha_2}
		\begin{bmatrix}
			0&0&0&0\\
			0&\alpha_1&0&-\beta_1\\
			0&0&0&0\\
			0&-\beta_1&0&1-\alpha_1
		\end{bmatrix}
		+ \frac{\gamma-1}{2\alpha_1\alpha_2}
		\begin{bmatrix}
			0&0&0&0\\
			0&0&0&0\\
			0&0&\alpha_2&-\beta_2\\
			0&0&-\beta_2&1-\alpha_2
		\end{bmatrix}.
	\end{align*}
	When $t_{ij}$ denotes the $(i,j)$-th entry of $T_1+T_2$, 
	it follows that $t_{11}=1/2$, $t_{12}=t_{21}=t_{13}=t_{31}=0$, $t_{14}=t_{41}=-\xi/2$, 
	$t_{23}=t_{32}=\xi/2$, 
	\begin{equation*}
		t_{24} = t_{42} = \frac{(\gamma-1)\beta_1}{2\alpha_1\alpha_2} - \frac{(\gamma-1)\beta_1}{2\alpha_1\alpha_2} = 0,\quad
		t_{34} = t_{43} = \frac{(\gamma-1)\beta_2}{2\alpha_1\alpha_2} - \frac{(\gamma-1)\beta_2}{2\alpha_1\alpha_2} = 0,
	\end{equation*}
	\begin{align*}
		t_{22} &= \frac{\beta_1^2}{2\alpha_1\alpha_2} + \frac{(\gamma-1)\alpha_1}{2\alpha_1\alpha_2}
		= \frac{1-\alpha_1}{2\alpha_2} + \frac{\gamma-1}{2\alpha_2} = 1/2,\\
		t_{33} &= \frac{\beta_2^2}{2\alpha_1\alpha_2} + \frac{(\gamma-1)\alpha_2}{2\alpha_1\alpha_2}
		= \frac{1-\alpha_2}{2\alpha_1} + \frac{\gamma-1}{2\alpha_1} = 1/2,\\
		t_{44} &= \frac{\xi^2}{2} + \frac{(\gamma-1)^2}{2\alpha_1\alpha_2}
		+ \frac{(\gamma-1)(1-\alpha_1)}{2\alpha_1\alpha_2} + \frac{(\gamma-1)(1-\alpha_2)}{2\alpha_1\alpha_2}\\
		&= \frac{(1-\alpha_1)(1-\alpha_2)}{2\alpha_1\alpha_2} + \frac{(\gamma-1)^2}{2\alpha_1\alpha_2}
		+ \frac{(\gamma-1)(2-\gamma)}{2\alpha_1\alpha_2}\\
		&= \frac{1-\gamma+\alpha_1\alpha_2}{2\alpha_1\alpha_2} + \frac{\gamma-1}{2\alpha_1\alpha_2} = 1/2.
	\end{align*}
	Therefore, 
	\[
	T_1+T_2 + \Gamma(T_1+T_2)
	= \frac{1}{2}
	\begin{bmatrix}
		1&0&0&-\xi\\
		0&1&\xi&0\\
		0&\xi&1&0\\
		-\xi&0&0&1
	\end{bmatrix}
	+ \frac{1}{2}
	\begin{bmatrix}
		1&0&0&\xi\\
		0&1&-\xi&0\\
		0&-\xi&1&0\\
		\xi&0&0&1
	\end{bmatrix}
	= I.
	\]
	\par
	Proof of (ii). First, assume $\gamma=1$.
	Then $t\in[0,1]$ and \eqref{supp-mat-S2} implies that 
	\[
	\alpha_1\alpha_2 \overset{\gamma=1}{=} (1-\alpha_2)(1-\alpha_1)
	\overset{\text{\eqref{supp-mat-S2}}}{\le} t\alpha_1\alpha_2
	\overset{t\in[0,1]}{\le} \alpha_1\alpha_2,
	\]
	whence $t=1$ and $s=1/2$.
	Since it is easily checked that $T_i\in\cK_{1/2}^{(0)}$ for all $i=1,2$, 
	we obtain (ii).
	Next, assume $\gamma>1$. 
	Since the function $\sqrt{t'}/(1+t')$, $t'\in[0,1]$, is increasing, 
	from \eqref{supp-mat-eq01'} and \eqref{supp-mat-eq01''}, 
	it follows that $\sco(v_2)=\sco(v_3)=0$ and 
	\[
	\sco(v_1) = \frac{\beta_1\beta_2/\alpha_1\alpha_2}{1 + (\beta_1\beta_2/\alpha_1\alpha_2)^2}
	\overset{\text{\eqref{supp-mat-S2}}}{\le} \frac{\sqrt{t}}{1+t} = s.
	\]
	Thus $T_1\in\cK_s^{(0)}$.
	Thanks to \eqref{supp-mat-local-uni}, we also have $T_2\in\cK_s^{(0)}$.
	Therefore, (ii) holds.
	\par
	Proof of (iii). First, assume $\gamma=1$.
	Then it is easily checked that $\Tr\rho_1T_2=0$ and 
	\[
	\Tr\rho_2T_1 = (1-\alpha_1)(1-\alpha_2) + \alpha_1\alpha_2 - 2\beta_1\beta_2
	\overset{\gamma=1}{=} \alpha_1\alpha_2 + \alpha_1\alpha_2 - 2\alpha_1\alpha_2
	= 0,
	\]
	which are just (iii).
	Next, assume $\gamma>1$.
	Since the equations \eqref{supp-mat-eq01}, $\rho_2\ket{v_2}=\rho_2\ket{v_3}=0$, and 
	\begin{align*}
		\braket{v_1 |\rho_2| v_1}
		&= (1-\alpha_1)(1-\alpha_2)
		+ \Bigl( \frac{\beta_1\beta_2}{\alpha_1\alpha_2} \Bigr)^2\alpha_1\alpha_2
		- \frac{\beta_1\beta_2}{\alpha_1\alpha_2}\cdot2\beta_1\beta_2\\
		&= (1-\alpha_1)(1-\alpha_2) - \frac{(\beta_1\beta_2)^2}{\alpha_1\alpha_2}
		= 0
	\end{align*}
	hold, we have 
	\[
	\Tr\rho_2T_1 = \frac{1}{2}\braket{v_1 |\rho_2| v_1}
	+ \frac{\gamma-1}{2\gamma}\braket{v_2 |\rho_2| v_2}
	+ \frac{\gamma-1}{2\gamma}\braket{v_3 |\rho_2| v_3}
	= 0.
	\]
	Moreover, since $\rho_2 = (U_A\otimes U_B)^\dag\rho_1(U_A\otimes U_B)$ holds, 
	we obtain 
	\begin{align*}
		\Tr\rho_1T_2
		&\overset{\text{\eqref{supp-mat-local-uni}}}{=} \Tr\rho_1(U_A\otimes U_B)T_1(U_A\otimes U_B)^\dag\\
		&= \Tr(U_A\otimes U_B)^\dag\rho_1(U_A\otimes U_B)T_1
		= \Tr\rho_2T_1 = 0.
	\end{align*}
	Therefore, (iii) holds.
\end{proof}

\begin{proof}[Proof of Lemma~$\ref{supp-mat-lem1}$]
	Assume that $s\in[0,1/2]$ and \eqref{supp-mat-S1}.
	All we need is to show that 
	\begin{enumerate}
		\item
		$T_1+T_2 + \Gamma(T_1+T_2) = I$,
		\item
		$\nege(T_i+\Gamma(T_i)) \le s$ for all $i=1,2$,
		\item
		$\Tr\rho_1T_2 = \Tr\rho_2T_1 = 0$,
	\end{enumerate}
	by the same reason as the proof of Lemma~\ref{supp-mat-lem1}.
	Since (i) and (iii) have been already proved, 
	we show only (ii).
	Also, the inequality $\gamma=\alpha_1+\alpha_2\ge1$ holds, 
	and we may assume $\alpha_1\alpha_2>0$, 
	by the same reason as the proof of Lemma~\ref{supp-mat-lem1}.
	\par
	Proof of (ii). First, assume $\gamma=1$.
	Then $s\in[0,1/2]$ and \eqref{supp-mat-S1} implies that 
	\[
	\alpha_1\alpha_2 \overset{\gamma=1}{=} (1-\alpha_2)(1-\alpha_1)
	\overset{\text{\eqref{supp-mat-S1}}}{\le} 4s^2\alpha_1\alpha_2
	\overset{s\in[0,1/2]}{\le} \alpha_1\alpha_2,
	\]
	whence $s=1/2$.
	Since it is easily checked that $\nege(T_i+\Gamma(T_i))=1/2$ for all $i=1,2$, 
	we obtain (ii).
	Next, assume $\gamma>1$. Then 
	\begin{align*}
		\nege(T_1+\Gamma(T_1)) &\le \nege(\Gamma(T_1))
		\overset{\text{\eqref{supp-mat-eq01''}}}{\le} \nege\Bigl( \Gamma\Bigl( \frac{1}{2}\ketbra{v_1}{v_1} \Bigr) \Bigr)\\
		&\overset{\text{\eqref{sc-neg}}}{=} \frac{1}{2}\|v_1\|^2\sco(v_1)
		\overset{\text{\eqref{supp-mat-eq01'}}}{=} \frac{\beta_1\beta_2}{2\alpha_1\alpha_2}
		\overset{\text{\eqref{supp-mat-S1}}}{\le} s.
	\end{align*}
	Since \eqref{supp-mat-local-uni} holds, 
	we have $\nege(T_2+\Gamma(T_2)) \le \nege(\Gamma(T_2)) = \nege(\Gamma(T_1)) \le s$.
	Therefore, (ii) holds.
\end{proof}


\section*{References}
\begin{thebibliography}{99}

\bibitem{Janotta2013}
P.~Janotta and R.~Lal, 
\textit{Phys. Rev.\ A} \textbf{87}, 052131 (2013).

\bibitem{Janotta2014}
P.~Janotta and H.~Hinrichsen, 
\textit{J.\ Phys.\ A} \textbf{47}, 323001 (2014).

\bibitem{Yoshida2018}
Y.~Yoshida and M.~Hayashi, 
arXiv:1801.03988. [\textit{J.\ Phys.\ A} (to be published)]

\bibitem{Bae2016}
J.~Bae, D.~G.~Kim, and L.~Kwek, \textit{Entropy}, \textbf{18}, 39 (2016).

\bibitem{Matsumoto2018}
K.~Matsumoto and G.~Kimura, arXiv:1802.01162.

\bibitem{Short2010}
A.~J.~Short and S.~Wehner, 
\textit{New J.\ Phys.\ }\textbf{12}, 033023 (2010).

\bibitem{Barnum2010}
H.~Barnum, J.~Barrett, L.~O.~Clark, M.~Leifer, R.~Spekkens, N.~Stepanik, A.~Wilce, and R.~Wilke, 
\textit{New J.\ Phys.\ }\textbf{12}, 033024 (2010).

\bibitem{Dahlsten2012}
O.~C.~O.~Dahlsten, D.~Lercher, and R.~Renner, 
\textit{New J.\ Phys.\ }\textbf{14}, 063024 (2012).

\bibitem{Muller2012}
M.~P.~M\"{u}ller and C.~Ududec, \textit{Phys.\ Rev.\ Lett.\ } \textbf{108}, 130401 (2012).

\bibitem{Muller2013}
M.~P.~M\"{u}ller and L.~Masanes, 
\textit{New J.\ Phys.\ }\textbf{15}, 053040 (2013).

\bibitem{Masanes2011}
L.~Masanes and M.~P.~M\"{u}ller, \textit{New J.\ Phys.\ } \textbf{13}, 063001 (2011).

\bibitem{Richens2017}
J.~G.~Richens, J.~H.~Selby, and S.~W.~Al-Safi, 
\textit{Phys.\ Rev.\ Lett.\ }\textbf{119}, 080503 (2017).

\bibitem{Lami2017}
L.~Lami, C.~Palazuelos, and A.~Winter, 
\textit{Comm.\ Math.\ Phys.\ }\textbf{361}, 661 (2018).

\bibitem{Aubrun2018}
G.~Aubrun, L.~Lami, C.~Palazuelos, S.~J.~Szarek, and A.~Winter, 
arXiv:1809.10616.

\bibitem{Lee2015}
C.~M.~Lee and J.~Barrett, \textit{New J.\ Phys.\ }\textbf{17}, 083001 (2015). 

\bibitem{PR1994}
S.~Popescu and D.~Rohrlich, 
\textit{Found.\ Phys.\ }\textbf{24}, 379 (1994).

\bibitem{Plavala2017}
M.~Plavala and M.~Ziman, 
arXiv:1708.07425.

\bibitem{Arai2019}
H.~Arai, Y.~Yoshida, and M.~Hayashi, 
\textit{J.\ Phys.\ A} \textbf{52}, 465304 (2019).


\bibitem{HayashiBook:2017}
M.~Hayashi, 
\textit{Quantum Information Theory Mathematical Foundation},
(Springer, Berlin, 2017), Second Edition.







\end{thebibliography}
\end{document}